\documentclass[1pt,a4paper]{article}
\usepackage[english]{babel}
\usepackage{color}
\usepackage[english]{babel}
\usepackage{color}
\usepackage{graphicx}
\usepackage{hyperref}
\usepackage{amsmath,amssymb,amsthm}
\usepackage[T1]{fontenc}
\usepackage[utf8]{inputenc}
\usepackage{authblk}
\usepackage[english]{babel}
\usepackage{amsmath,amssymb,amsthm}
\usepackage{color}
\usepackage{graphics}
\usepackage{amsmath,amssymb,amsthm}
\newtheorem{definition}{Definition}
\newtheorem{theorem}[definition]{Theorem}
\newtheorem{lemma}[definition]{Lemma}

\newtheorem{remark}[definition]{Remark}
\newtheorem{proposition}[definition]{Proposition}

\renewcommand{\tilde}[1]{\widetilde{#1}}

\title{Further improvements on the Feng-Rao bound for dual codes}
\author[1]{Olav Geil\thanks{olav@math.aau.dk}}
\author[1,2]{Stefano Martin\thanks{stefano@math.aau.dk}}
\affil[1]{Department of Mathematical Sciences, Aalborg University}
\affil[2]{Engineering Software Institute, East China Normal University}

\begin{document}
\maketitle
\begin{abstract}
Salazar, Dunn and Graham in~\cite{salazar} presented an improved Feng-Rao
bound for the minimum distance of dual codes. In this work we
take the improvement a step further. Both the original bound by Salazar et.\ al., as well as our
improvement are lifted so that they deal with generalized Hamming
weights. We also demonstrate the advantage of working with one-way
well-behaving pairs rather than weakly well-behaving or well-behaving pairs.\\


\end{abstract}

\section{Introduction}\label{secone}
The celebrated Feng-Rao bound for the minimum distance of dual codes
\cite{FR1,FR2} was originally presented in a language close to that of
affine variety codes \cite{lax}. A more general result was derived by
formulating the bound at the level of general linear codes
\cite{early,miura1,MM,geithom}. 
Among the general linear code
formulations the weakest version uses one basis for ${\mathbb{F}}_q^n$
and the concept of {\textit{well-behaving pairs}} (WB). The stronger versions use
two or even three bases and the concept of {\textit{weakly well-behaving}} (WWB)
or even {\textit{one-way well-behaving}} (OWB). The strong linear code formulation is the most general
of all versions of the Feng-Rao bound in the sense that all
other formulations, including the order bound~\cite{handbook}, can be viewed as
corollaries to it.\\
 
In~\cite{salazar} Salazar, Dunn and Graham presented a
clever improvement to the Feng-Rao bound for the minimum distance of
dual codes which they name {\textit{the advisory bound}}~\cite[Def.\ 40]{salazar}. Their exposition uses a language close to that of Feng and
Rao's original papers. In the present paper we start by giving a
general linear code enhancement of their bound and we lift it to deal
with generalized Hamming weights improving upon the usual Feng-Rao
bound for generalized Hamming weights of dual codes
\cite{heijnenpellikaan,geithom}. 
We remind the reader
that generalized Hamming weights among other things are relevant for the analysis of
wiretap channels of type II \cite{wei,luo} and secret sharing
schemes based on error correcting codes \cite{kurihara}. Our proof
demonstrates that the advisory bound is a consequence of a lemma
from which further improvements can be derived. These improvements are
investigated in detail and are formulated in a separate bound. The new 
bound is then lifted to deal with generalized Hamming weights. Our
exposition involves as a main ingredient a relaxation of the concept of OWB.\\
The paper~\cite{salazar} describes two families of affine variety codes
for which the advisory bound is sometimes strictly better than the Feng-Rao
bound. The first family~\cite[Sec.\ 3.1]{salazar} is related to a
curve over ${\mathbb{F}}_8$. The second family~\cite[Sec.\
3.2]{salazar} relates to a surface over ${\mathbb{F}}_4$. In
Section~\ref{secexamples} we shall give a thorough treatment of the
curve from~\cite[Sec.\ 3.1]{salazar} and a related curve over
${\mathbb{F}}_{27}$. As it shall be demonstrated for these curves
sometimes the new bound produces much better results than the advisory
bound. Also it is demonstrated for the first time in the literature
that the Feng-Rao bound equipped with OWB can sometimes be much better
than the same bound equipped with WWB. We do not treat the surface from~\cite[Sec.\
3.2]{salazar} in the present paper. This is due to the fact that it is
more natural to treat the corresponding quotient ring as an order
domain with weights in ${\mathbb{N}}_0^2$ \cite{GP,AG}. Doing so,
one finds much better code parameters by applying the usual Feng-Rao
bound than what was produced by the advisory bound in~\cite[Sec.\
3.2]{salazar}. It is beyond the scope of the present paper to give the details.
\section{Enhancements of the advisory bound}\label{sectwo}
To explain better what is the essence of Salazar,
Dunn, and Graham's method, below we explain it at the level of
general linear codes. We also extend their method to deal with
generalized Hamming weights.\\

Let $n$ be a positive integer and $q$ a prime power. Throughout this and the following section we consider a fixed ordered triple
$({\mathcal{U}},{\mathcal{V}},{\mathcal{W}})$ where
${\mathcal{U}}=\{\vec{u}_1, \ldots ,
\vec{u}_n\}$, ${\mathcal{V}}=\{\vec{v}_1, \ldots , \vec{v}_n\}$, and
${\mathcal{W}}=\{\vec{w}_1, \ldots , \vec{w}_n\}$ are three (possibly
different)  bases for
${\mathbb{F}}_q^n$ as a vector space over ${\mathbb{F}}_q$. By ${\mathcal{I}}$ we
shall always mean the set $\{1, \ldots , n\}$.

\begin{definition}\label{def1}
Let the 
function $\bar{\rho}_{\mathcal{W}}: {\mathbb{F}}_q^n \rightarrow \{0, 1, \ldots ,
n\}$ be given as follows. For $\vec{c} \neq \vec{0}$ we let $\bar{\rho}_{\mathcal{W}}(\vec{c})=i$ if $\vec{c} \in
{\mbox{Span}}\{\vec{w}_1, \ldots , \vec{w}_i\}  \backslash
{\mbox{Span}} \{\vec{w}_1, \ldots , \vec{w}_{i-1}\}$. Here, we used
the notion ${\mbox{Span}}\,  \emptyset = \{\vec{0}\}$. Finally, we let 
$\bar{\rho}_{\mathcal{W}}(\vec{0})=0$.
\end{definition}
The following two concepts play a crucial role in our exposition.
\begin{definition}\label{defcomp}
The component wise product of two vectors $\vec{u}$ and $\vec{v}$ in
${\mathbb{F}}_q^n$ is defined by $(u_1, \ldots , u_n)\ast (v_1,
\ldots, v_n)=(u_1v_1, \ldots , u_nv_n)$.
\end{definition}
\begin{definition}\label{def3}
Let an ordered triple of bases $({\mathcal{U}},{\mathcal{V}},{\mathcal{W}})$ be given.
We define $m : {\mathbb{F}}_q^n \backslash \{ \vec{0}\} \rightarrow
{\mathcal{I}}$ by $m(\vec{c})=l$ if $l$ is the smallest number in
${\mathcal{I}}$ for which 
$\vec{c} \cdot \vec{w}_l \neq 0$.
\end{definition}
We start by stating the Feng-Rao bound for the minimum distance
of dual codes.
\begin{definition}\label{defweb}
Let $({\mathcal{U}},{\mathcal{V}},{\mathcal{W}})$ and ${\mathcal{I}}$ be as
above.\\
An ordered pair $(i,j)\in {\mathcal{I}}\times {\mathcal{I}}$ is said
to be well-behaving (WB) if
$\bar{\rho}_{\mathcal{W}}(\vec{u}_{i^\prime} \ast \vec{v}_{j^\prime})
<\bar{\rho}_{\mathcal{W}}(\vec{u}_{i} \ast \vec{v}_{j})$ holds for all
$i^\prime \leq i$ and $j^\prime \leq j$ with $(i^\prime
,j^\prime)\neq(i,j)$.\\
Less restrictive $(i,j)\in {\mathcal{I}}\times {\mathcal{I}}$ is said
to be weakly well-behaving (WWB) if
$\bar{\rho}_{\mathcal{W}}(\vec{u}_{i^\prime} \ast \vec{v}_{j})
<\bar{\rho}_{\mathcal{W}}(\vec{u}_{i} \ast \vec{v}_{j})$ and 
$\bar{\rho}_{\mathcal{W}}(\vec{u}_{i} \ast \vec{v}_{j^\prime})
<\bar{\rho}_{\mathcal{W}}(\vec{u}_{i} \ast \vec{v}_{j})$ hold for all
$i^\prime < i$ and $j^\prime < j$.\\
Even less restrictive $(i,j)\in {\mathcal{I}}\times {\mathcal{I}}$ is said
to be one-way well-behaving (OWB) if
$\bar{\rho}_{\mathcal{W}}(\vec{u}_{i^\prime} \ast \vec{v}_{j})
<\bar{\rho}_{\mathcal{W}}(\vec{u}_{i} \ast \vec{v}_{j})$ 
holds for all
$i^\prime < i$.
\end{definition}
The usual Feng-Rao bound for the minimum distance of dual codes reads.
\begin{theorem}
For $\vec{c} \in {\mathbb{F}}_q^n\backslash \{ \vec{0}\}$ write
$l=m(\vec{c})$. The Hamming weight of $\vec{c}$ satisfies
\begin{eqnarray}
w_H(\vec{c})&\geq& \# \{ (i,j) \in {\mathcal{I}} \times {\mathcal{I}}
\mid \bar{\rho}_{\mathcal{W}}(\vec{u}_i \ast \vec{v}_j)=l {\mbox{ and
  }} (i,j) {\mbox{ is OWB}}\} \label{eqfrowb}\\
&\geq& \# \{ (i,j) \in {\mathcal{I}} \times {\mathcal{I}}
\mid \bar{\rho}_{\mathcal{W}}(\vec{u}_i \ast \vec{v}_j)=l {\mbox{ and
  }} (i,j) {\mbox{ is WWB}}\} \label{eqfrwwb}\\
&\geq& \# \{ (i,j) \in {\mathcal{I}} \times {\mathcal{I}}
\mid \bar{\rho}_{\mathcal{W}}(\vec{u}_i \ast \vec{v}_j)=l {\mbox{ and
  }} (i,j) {\mbox{ is WB}}\}. \label{eqfrwb}
\end{eqnarray}
\end{theorem}
From~\cite[Ex.\ 2.6]{MM} and \cite[Sec.\ 3.1]{salazar} we have
examples where (\ref{eqfrwwb}) are stronger than
(\ref{eqfrwb}). Section~\ref{secexamples} demonstrates that also (\ref{eqfrowb}) can be stronger than (\ref{eqfrwwb}). This
fact was not known before.\\
Although~\cite{salazar} considered only WB and WWB we shall state our
enhancement of the advisory bound using OWB. Doing so we get the
strongest possible version which in addition requires the minimal
number of calculations.
\begin{definition}\label{def2}
Let $({\mathcal{U}},{\mathcal{V}},{\mathcal{W}})$ and ${\mathcal{I}}$ be as
above. 
Consider ${\mathcal{I}}^\prime=\{i_1, \ldots , i_s\}
\subseteq {\mathcal{I}}$ with $i_a \neq i_b$ for $a \neq b$. An ordered pair $(i,j) \subseteq
{\mathcal{I}}^\prime \times {\mathcal{I}}$ is said to be one-way well-behaving (OWB) with
respect to ${\mathcal{I}}^\prime$ if $\bar{\rho}_{\mathcal{W}}(\vec{u}_{i^\prime} \ast
\vec{v}_j) < \bar{\rho}_{\mathcal{W}}(\vec{u}_{i} \ast
\vec{v}_j)$ holds for all $i^\prime \in {\mathcal{I}}^\prime$ with $i^\prime <
i$. \\
We say that ${\mathcal{I}}^\prime$ has the $\mu$-property with respect
to  $l$
if for all $i \in {\mathcal{I}}^\prime$ there exists a  $j \in {\mathcal{I}}$ such that
\begin{enumerate}
\item $(i,j)$ is OWB with respect to ${\mathcal{I}}^\prime$,
\item $\bar{\rho}_{\mathcal{W}}(\vec{u}_i \ast \vec{v}_j)=l$.
\end{enumerate}
\end{definition}
The following theorem is an enhancement of the advisory bound
\cite[Th.\ 48]{salazar}.  
\begin{theorem}\label{thehere}
Let $\vec{c} \in {\mathbb{F}}_q^n \backslash \{ \vec{0}\}$. We have
\begin{eqnarray}
w_H(\vec{c}) & \geq & \max \{\# {\mathcal{I}}^\prime \mid {\mathcal{I}}^\prime \subseteq {\mathcal{I}},
{\mathcal{I}}^\prime {\mbox{ \ has the $\mu$-property with respect to \ }} m(\vec{c}) \}. \nonumber 
\end{eqnarray}
\end{theorem}
\begin{proof}
The theorem  is a special case of
Theorem~\ref{thebelow} below.
\end{proof}

\begin{remark}\label{remstart}
Consider the code $C(s)=\{\vec{c} \in {\mathbb{F}}_q^n \mid \vec{c}
\cdot \vec{w}_1=\cdots = \vec{c} \cdot \vec{w}_s=0\}$. To estimate the
minimum distance of $C(s)$ we calculate the minimal value from
Theorem~\ref{thehere} when $m(\vec{c})$ runs through all possible
numbers in $\{s+1, \ldots , n\}$. As an alternative to $C(s)$ we get
an improved code construction
by using as parity checks only those $\vec{w}_l$, $l \in {\mathcal{I}}$ for which
Theorem~\ref{thehere} with $m(\vec{c})=l$ produces values less than
$\delta$. The minimum distance of this code, which we denote by
$\tilde{C}_{adv}(\delta)$, is at least $\delta$.
\end{remark}
We next consider the generalized Hamming weights.
\begin{definition}
Let $C \subseteq {\mathbb{F}}_q^n$ be a code of dimension $k$. For
$t=1, \ldots , k$ the $t$th generalized Hamming weight is
$$d_t(C)=\min \{ \# {\mbox{Supp}} \, D \mid D {\mbox{ \ is a subspace of
    \ }} C {\mbox{ \ of dimension \ }}
t\}.$$
\end{definition}
Clearly, $d_1$ is nothing but the usual minimum
distance. To estimate generalized Hamming weights we first need to extend
Definition~\ref{def2} and Definition~\ref{def3}.
\begin{definition}\label{def4}
Consider $1 \leq l_1 < \cdots < l_t \leq n$ and let ${\mathcal{I}}^\prime
\subseteq {\mathcal{I}}$. We will say that ${\mathcal{I}}^\prime$ has the $\mu$-property with
respect to $\{l_1, \ldots , l_t\}$ if for all $i \in {\mathcal{I}}^\prime$ there
exists a $j\in {\mathcal{I}}$ such that
\begin{itemize}
\item $(i,j)$ is OWB with respect to ${\mathcal{I}}^\prime$,\item $\bar{\rho}_{\mathcal{W}}(\vec{u}_i \ast \vec{v}_j)\in \{l_1, \ldots , l_t\}$.
\end{itemize}
\end{definition}

\begin{definition}\label{def5}
Let $D \subseteq {\mathbb{F}}_q^n$ be a subspace. We
define
$$m(D)=\big\{ m(\vec{c}) \mid \vec{c} \in D \backslash \{ \vec{0}\}\big\}.$$
\end{definition}
The following proposition is easily proved.
\begin{proposition}\label{pro1}
If $D \subseteq {\mathbb{F}}_q^n$ is a subspace of dimension $t$ then
$\# m(D) =t$.
\end{proposition}
Our enhancement of the advisory bound is based on the
following lemma from which we shall also in the next section derive an
even better bound.
\begin{lemma}\label{lemlommelaerke}
Consider a subspace $D \subseteq {\mathbb{F}}_q^n$. Let $U \subseteq
{\mathbb{F}}_q^n$ be a subspace of dimension $\delta$ such that for
all non-zero words $\vec{u} \in U$ for some $\vec{v}_j \in {\mathcal{V}}$ and
some $\vec{c} \in D$ it holds that $(\vec{u} \ast \vec{v}_j) \cdot
\vec{c} \neq 0$ then $|{\mbox{Supp}} \, D| \geq \delta$. 
\end{lemma}
\begin{proof}
Aiming for a contradiction we assume that the above criteria holds
true, but that $| {\mbox{Supp}}\, D | <\delta$. Without loss of
generality we write ${\mbox{Supp}} \, D =\{ 1, \ldots , g\}$. Clearly
$g \leq \delta -1$. Consider a matrix whose rows constitute a basis
for $U$. After having performed Gaussian elimination we arrive at a
matrix whose last row, say
$\vec{u}^\prime$, starts with $\delta
-1$ zeros. Therefore $\vec{u}^\prime \ast \vec{c} =\vec{0}$ holds for all
$\vec{c} \in D$. On the other hand by assumption for some particular word $\vec{c} \in D$ we have $(\vec{u}^\prime
\ast \vec{v}_j)\cdot \vec{c} \neq 0 \Rightarrow \vec{u}^\prime \ast \vec{c}
\neq \vec{0}$. This is a contradiction. 
\end{proof}
\begin{theorem}\label{thebelow}
Consider a subspace $D \subset {\mathbb{F}}_q^n$. We have
\begin{eqnarray}
\# {\mbox{Supp}}\, D&\geq &\max \{ \# {\mathcal{I}}^\prime \mid {\mathcal{I}}^\prime \subseteq {\mathcal{I}},
{\mathcal{I}}^\prime {\mbox{ \ has the }} \nonumber \\
&&{\mbox{ \ \ \ \ \ \ \ \ \ $\mu$-property with respect to\ }} m(D)\}. \nonumber
\end{eqnarray}
\end{theorem}
\begin{proof}
Let ${\mathcal{I}}^\prime=\{ i_1, \ldots , i_\delta\}$, $i_a \neq i_b$ for $a \neq
b$, be a set which has the
$\mu$-property with respect to $m(D)$. Consider $\sum_{r=1}^s
\alpha_r\vec{u}_{i_r}$, $1 \leq s \leq \delta$, $\alpha_r \in
{\mathbb{F}}_q$,  $\alpha_s \neq 0$. By
assumption there exists a $j \in  {\mathcal{I}}$ such that $(i_s,j)$ is OWB with respect
to ${\mathcal{I}}^\prime$ and such that
$\bar{\rho}_{\mathcal{W}}(\vec{u}_{i_s} \ast \vec{v}_j) \in
m(D)$. Therefore, $\bar{\rho}_{\mathcal{
W}}\big(\big(\sum_{r=1}^s
\alpha_r\vec{u}_{i_r}\big)\ast \vec{v}_j\big) \in m(D)$ and for some
$\vec{c} \in D$ it holds that $\big(\sum_{r=1}^s
\alpha_r\vec{u}_{i_r}\big)\ast \vec{v}_j\big)\cdot \vec{c} \neq 0$. The theorem
now follows from Lemma~\ref{lemlommelaerke}.
\end{proof}
\begin{remark}\label{remm}
Let $\{\vec{d}_1, \ldots , \vec{d}_{n-k}\} \subseteq {\mathbb{F}}_q^n$ be a linearly independent
set and consider the code 
$C=\{ \vec{c}\in {\mathbb{F}}_q^n \mid \vec{c} \cdot
\vec{d}_1=\cdots = \vec{c} \cdot
\vec{d}_{n-k}=0\}$. 
Without loss of generality we may assume that
$\bar{\rho}_{\mathcal{W}}(\vec{d}_1)<\cdots <
\bar{\rho}_{\mathcal{W}}(\vec{d}_{n-k})$ holds, say these numbers are
$l_1 <\cdots <l_{n-k}$. It is not hard to prove that
$m(C)={\mathcal{I}}\backslash \{ l_1, \ldots , l_{n-k}\}$.
\end{remark}
Combining Theorem~\ref{thebelow} and Remark~\ref{remm} we get:
\begin{theorem}\label{deter9}
Let $C=\{ \vec{c}\in {\mathbb{F}}_q^n \mid \vec{c} \cdot
\vec{d}_1=\cdots = \vec{c} \cdot
\vec{d}_{n-k}=0\}$, where $\{\vec{d}_1, \ldots , \vec{d}_{n-k} \}$ and $\{l_1,
\ldots , l_{n-k}\}$ are as in Remark~\ref{remm}. 
For $t=1, \ldots ,k$ the $t$th generalized Hamming weight of $C$ satisfies
\begin{eqnarray}
d_t(C)\geq  \min \bigg\{ {\mbox{ \ \hspace{9.9cm}}} \nonumber \\
 \max \big\{ \# {\mathcal{I}}^\prime \mid {\mathcal{I}}^\prime \subseteq {\mathcal{I}},
{\mathcal{I}}^\prime {\mbox{ \ has the $\mu$-property with respect to \ }} \{m_1,
\ldots , m_t\}\big\} \mid {\mbox{\ \hspace{0.1cm}}} \nonumber \\
{\mbox{ \ \hspace{16mm}}} m_1 < \cdots < m_t, m_s \in {\mathcal{I}}\backslash
\{l_1, \ldots , l_{n-k} \} {\mbox{ \  for \ }} s=1, \ldots , t \bigg\}.
\nonumber 
\end{eqnarray}
\end{theorem}
In Section~\ref{secexamples} we illustrate with a couple of examples
that Theorem~\ref{deter9} is operational
even though it does appear 
technical at a first glance.\\
In a straight forward manner one can enhance Theorem~\ref{deter9} to
also deal with relative generalized Hamming weights
(See~\cite{luo,luoetal}). This bound should be compared with the naive
bound, that the relative generalized Hamming weight is always at least
as large as the estimate on the generalized Hamming weight from
Theorem~\ref{deter9}. It should also be compared to the Feng-Rao bound
for relative generalized Hamming weights.
As we have no examples where
the mentioned enhancement of Theorem~\ref{deter9} produces results
which are simultaneously better than the above mentioned two
alternatives and as at the same time the enhancement of
Theorem~\ref{deter9} is rather technical we do not give the details here.
\section{Further improvements}\label{secfurther}
In the following we will strengthen the results from the previous
section. We start by explaining how to improve upon
Theorem~\ref{thehere}. Given $\vec{c} \in {\mathbb{F}}_q^n \backslash
\{ \vec{0} \}$, consider the corresponding number $m(\vec{c})=\min \{l
  \mid \vec{c} \cdot \vec{w}_l \neq 0\}$ and a set
  ${\mathcal{I}}^\prime \subseteq {\mathcal{I}}$ which has the
  $\mu$-property with respect to $m(\vec{c})$. Theorem~\ref{thehere}
  relies on the observation that if for $i \in {\mathcal{I}}^\prime$, 
  $j \in {\mathcal{I}}$ is the corresponding number such that
  $\bar{\rho}_{\mathcal{W}}(\vec{u}_i\ast \vec{v}_j)=m(\vec{c})$ and $(i,j)$ is OWB with respect to
  ${\mathcal{I}}^\prime$ then 
$$\vec{c} \cdot \bigg( \big(\sum_{\begin{array}{c}i^\prime \in {\mathcal{I}}^\prime
    \\i^\prime \leq i
\end{array}}\alpha_{i^\prime}\vec{u}_{i^\prime}\big)\ast
\vec{v}_j\bigg) \neq 0$$
holds whenever $\alpha_{i^\prime} \in
    {\mathbb{F}}_q$, 
$\alpha_i \neq 0$. Note that the above argument uses no
information regarding the status of $\vec{c} \cdot
\vec{w}_{m(\vec{c})+1}, \cdots , \vec{c} \cdot
\vec{w}_{n}$. Indeed, if the only information we have on $\vec{c}$ is
$m(\vec{c})$ then these numbers can take on all possible combinations
of values from ${\mathbb{F}}_q$.
\begin{remark}\label{remms}
Let $C$ be as in Remark~\ref{remm} with
\begin{equation}
\bar{\rho}_{\mathcal{W}}(\vec{d}_1)=l_1 < \cdots <
\bar{\rho}_{\mathcal{W}}(\vec{d}_{n-k})=l_{n-k}.\label{eqved}
\end{equation}
Consider a general codeword $\vec{c} \in C \backslash \{\vec{0}
\}$. If the only thing we know about $\vec{d}_1, \ldots ,
\vec{d}_{n-k}$ is~(\ref{eqved}) then we have no information regarding
$\vec{c} \cdot \vec{w}_{l_1}, \ldots , \vec{c} \cdot
\vec{w}_{l_{n-k}}$. If however, as the other extreme,
we know that $\vec{d}_1=\vec{w}_{l_1}, \ldots , \vec{d}_{n-k}=\vec{w}_{l_{n-k}}$
then we have $\vec{c} \cdot \vec{w}_{l_1}= \cdots = \vec{c} \cdot
\vec{w}_{l_{n-k}}=0$.
\end{remark}
Write $l=m(\vec{c})$ and consider the indexes $l+1, \ldots , l+v\leq
n$. Here, $v$ is some positive integer. For some of the above indexes $x$ we may {\it{a priori}}
know that $\vec{c} \cdot \vec{w}_x =0$
(Remark~\ref{remms}). Let $l^\prime_1, \ldots , l^\prime_s$ be the
remaining indexes from $\{l+1, \ldots , l+v\}$. The idea in our
improvement to Theorem~\ref{thehere} is to consider separately the
following $s+1$ cases:
$$
\begin{array}{ll}
{\mbox{Case 0:}}&\vec{c} \cdot \vec{w}_{l^\prime_1}= \cdots = \vec{c}
\cdot \vec{w}_{l^\prime_s}=0.\\
{\mbox{Case 1:}}&\vec{c} \cdot \vec{w}_{l^\prime_1}\neq 0. \\
{\mbox{Case 2:}}&\vec{c} \cdot \vec{w}_{l^\prime_1}=0, \vec{c} \cdot
\vec{w}_{l^\prime_2}\neq 0.\\
 & {\mbox{ \ \ \hspace{2cm}}} \vdots \\
{\mbox{Case s:}}&\vec{c} \cdot \vec{w}_{l^\prime_1}= \cdots =\vec{c}
\cdot{\vec{w}}_{l^\prime_{s-1}}=0, \vec{c} \cdot
\vec{w}_{l^\prime_s}\neq 0.
\end{array}
$$
In each case $z$ we establish a set
${\mathcal{I}}^\prime_{z}\subseteq {\mathcal{I}}$ such that for every
non-zero linear combination $\sum_{i\in
  {\mathcal{I}}^\prime_z}\alpha_i \vec{u}_i$, $\alpha_i \in
{\mathbb{F}}_q$, 
a $\vec{v}_j \in {\mathcal{V}}$ exists with 
$$\vec{c} \cdot \bigg( \big(\sum_{i \in {\mathcal{I}}^\prime_z}\alpha_i
\vec{u}_i\big)\ast \vec{v}_j\bigg)  \neq 0.$$
From Lemma~\ref{lemlommelaerke} it then follows that $w_H(\vec{c})
\geq \min \{ \# {\mathcal{I}}^\prime_0, \ldots ,
\#{\mathcal{I}}^\prime_s\}$. The following definition is what we need to
  deal with the above set-up. We should stress that although
  Definition~\ref{definvolved} may appear long and technical, it is
  often quite manageable. This will be demonstrated in Section~\ref{secexamples}.
\begin{definition}\label{definvolved}
Consider the numbers $1 \leq l, l+1, \ldots , l+g \leq n$. 
A set 
${\mathcal{I}}^\prime\subseteq {\mathcal{I}}$ is said to have the $\mu$-property with
respect to $l$ with exception $\{l+1, \ldots , l+g\}$ if for all $i
\in {\mathcal{I}}^\prime$ a $j\in {\mathcal{I}}$ exists such that
\begin{itemize}
\item[(1a)] $\bar{\rho}_{\mathcal{W}}(\vec{u}_i \ast \vec{v}_j)=l$, and
\item[(1b)] for
  all $i^\prime \in {\mathcal{I}}^\prime$ with $i^\prime < i$ either $\bar{\rho}_{\mathcal{W}}(\vec{u}_{i^\prime}\ast
  \vec{v}_j)< l$ or $\bar{\rho}_{\mathcal{W}}(\vec{u}_{i^\prime}\ast
  \vec{v}_j) \in \{l+1, \ldots , l+g\}$ holds.
\end{itemize}
Assume next that $l+g+1 \leq n$. The set ${\mathcal{I}}^\prime$ is
said to have the relaxed $\mu$-property with respect to $(l,l+g+1)$
with exception $\{l+1, \ldots , l+g\}$ if for all $i \in
{\mathcal{I}}^\prime$ a $j \in {\mathcal{I}}$ exists such that either
conditions $(1a)$ and $(1b)$ above hold or 
\begin{itemize}
\item[(2a)] $\bar{\rho}_{\mathcal{W}}(\vec{u}_i \ast \vec{v}_j)=l+g+1$,
  and
\item[(2b)] $(i,j)$ is OWB with respect to ${\mathcal{I}}^\prime$, and
\item[(2c)] no $i^\prime \in {\mathcal{I}}^\prime$ with $i^\prime < i$
  satisfies $\bar{\rho}_{\mathcal{W}}(\vec{u}_{i^\prime} \ast \vec{v}_j)=l$.
\end{itemize}
\end{definition}
From the discussion above we arrive at the following improvement to Theorem~\ref{thehere}.
\begin{theorem}\label{thenew}\label{thestrongone}
Consider a non-zero codeword $\vec{c}$ and let $l=m(\vec{c})$. Choose
a non-negative integer $v$ such that $l+v\leq n$. Assume that for some
indexes $x \in \{l+1, \ldots , l+v\}$ we know {\textit{a priori}} that
$\vec{c} \cdot \vec{w}_x=0$. Let $l^\prime_1< \cdots  l^\prime_s$ be
the remaining indexes from $\{l+1, \ldots , l+v\}$. Consider the sets
${\mathcal{I}}^\prime_0, {\mathcal{I}}_1^\prime, \ldots ,
{\mathcal{I}}_s^\prime$ such that:
\begin{itemize}
\item ${\mathcal{I}}^\prime_0$ has the $\mu$-property with respect to
  $l$ with exception $\{ l+1, \ldots , l+v\}$.
\item For $i=1, \ldots , s$, ${\mathcal{I}}^\prime_i$ has the relaxed
  $\mu$-property with respect to $(l, l^\prime_i)$ with exception
  $\{l+1, \ldots , l^\prime_i -1\}$.
\end{itemize}
We have 
\begin{equation}
w_H(\vec{c}) \geq \min \{ \# {\mathcal{I}}_0^\prime, \#
{\mathcal{I}}_1^\prime, \ldots ,  \# {\mathcal{I}}_s^\prime\}. \label{eqcirkel} 
\end{equation}
To establish a lower bound on the minimum distance of a code $C$ we
repeat the above process for each $l\in m(C)$. For each such $l$ we 
choose a corresponding $v$, we determine sets ${\mathcal{I}}^{\prime}_i$ as
  above and we calculate the right side of~(\ref{eqcirkel}). The
  smallest value found constitutes a lower bound on the minimum distance.
\end{theorem}
\begin{remark}\label{remslut}
The results in Remark~\ref{remstart} also hold if we replace
Theorem~\ref{thehere} with Theorem~\ref{thenew}. We shall denote the
resulting improved codes by $\tilde{C}_{fim}(\delta)$ (here, {\textit{fim}} stands
for further improved).
\end{remark}
\begin{remark}
Assume ${\mathcal{I}}^\prime$ has the $\mu$-property with respect to
$l$. One possible choice of sets ${\mathcal{I}}^\prime_0,{\mathcal{I}}^\prime_1,\ldots ,
{\mathcal{I}}^\prime_s \subseteq {\mathcal{I}}$ in
Theorem~\ref{thestrongone} would be to choose all of them to be equal to
${\mathcal{I}}^\prime$. It follows that
Theorem~\ref{thestrongone} is indeed at least as strong as
Theorem~\ref{thehere}. The above observation relates to the fact that
Theorem~\ref{thenew} reduces to Theorem~\ref{thehere} when $v$ is
chosen to be always equal to $0$. 
\end{remark}

As shall be demonstrated later in the paper, Theorem~\ref{thenew} can
sometimes be much better than Theorem~\ref{thehere}. 
 For Theorem~\ref{thenew} to be operational we need a clever method to choose
for each $l \in m(C)$ the corresponding number $v$. As shall be clear form the examples in
Section~\ref{secexamples} for affine variety codes there is a very
natural way to do this. Another remark is that when the task is to
estimate the minimum distance of a fixed code, then we can set $v$
equal to $0$ for most values of $l$, reserving non-zero values to those $l$
for which Theorem~\ref{thehere} produces the smallest numbers. These
are the numbers that need to be improved.\\

In a similar way as Theorem~\ref{thehere} was enhanced to deal with
generalized Hamming weighs and relative generalized Hamming weights
we can enhance Theorem~\ref{thestrongone}. The notation in
Definition~\ref{definvolved} being already involved we only illustrate
how to deal with the second generalized Hamming weight. From that description
it should be clear how to deal with higher weights.
\begin{proposition}\label{propihop}
Let the notation be as in Theorem~\ref{thenew}. Consider a subspace $D
\subseteq C$ of dimension $2$, say $m(D)=\{a,b\}$. Let $v_a$ be the
$v$ corresponding to $l=a$. Let $a_1^\prime < \cdots <a^\prime_{s_a}$
be the numbers $l^\prime_1< \cdots < l^\prime_s$ corresponding to
$l=a$. Analogously for the case b. Referring to Definition~\ref{definvolved}, for
$\alpha=1, \ldots , s_a$ and $\beta = 1, \ldots , s_b$ we define
subsets of ${\mathcal{I}}$ as follows:
\begin{itemize}
\item ${\mathcal{I}}^{\prime \prime}_{0,0}$ is a set such that for all
  $i \in {\mathcal{I}}^{\prime \prime}_{0,0}$ for an $l \in
  \{a,b\}$ a $j$ exists such that (1a) and (1b) hold with
  $g=v_a$ if $l=a$, and $g=v_b$ if $l=b$.
\item  ${\mathcal{I}}^{\prime \prime}_{\alpha,0}$ is a set such that
  for all $i \in {\mathcal{I}}^{\prime \prime}_{\alpha,0}$ a $j$
  exists such that one of the following two conditions holds:
\begin{itemize}
\item Either (1a), (1b) or (2a), (2b), (2c) hold with $l=a$ and
  $g+1=a^\prime_\alpha$.
\item (1a) and (1b) hold with $l=b$ and $g=v_b$.
\end{itemize} 
\item ${\mathcal{I}}^{\prime \prime}_{0,\beta}$ is defined similarly
  to ${\mathcal{I}}^{\prime \prime}_{\alpha,0}$.
\item  ${\mathcal{I}}^{\prime \prime}_{\alpha,\beta}$ is a set such
  that for all $i \in  {\mathcal{I}}^{\prime \prime}_{\alpha,\beta}$
  an $l \in \{ a,b\}$ and a $j \in
  {\mathcal{I}}$ exist such that either (1a), (1b) or (2a), (2b), (2c)
  hold. Here, $g+1= a^\prime_\alpha$ if $l=a$, and $g+1=b^\prime_\beta$ if
  $l=b$.
\end{itemize}
The support of $D$ is of size at least equal to the smallest
cardinality of the above sets.
To establish a lower bound on the second generalized Hamming weight of a code $C$ we
repeat the above process for each $(a,b)\in m(C) \times m(C)$ with $a <b$. The
  smallest value found constitutes a lower bound on the second
  generalized Hamming weight.
\end{proposition}
Applying in larger generality the method described in the above proposition we derive lower
bounds on any generalized Hamming weights of $C$. 
It is clear that
this method can be of much higher complexity than the method described in
Theorem~\ref{deter9}. To lower the complexity we choose (referring to
the case of the second weight) most $v_a$ and
$v_b$ equal to zero, reserving non-zero values to those $(a,b)$ for
which Theorem~\ref{deter9} produces low
values. As shall be demonstrated in
the following section, Proposition~\ref{propihop} and its
generalization to higher weights can sometimes produce
much better results than Theorem~\ref{deter9}.\\
Similar results on the relative generalized Hamming weights as those
mentioned at the end of Section~\ref{sectwo} hold for the method
described above.

\section{Examples}\label{secexamples}
In this section we apply the advisory bound and the improved bound
from Section~\ref{secfurther} to affine variety codes coming from two particular
curves. The first curve corresponds to~\cite[Sec.\ 3.1]{salazar}. It
is a plane curve over ${\mathbb{F}}_8$. The second curve is the
natural counterpart for the field ${\mathbb{F}}_{27}$. We shall need a
couple of results from Gr\"{o}bner basis theory.\\
\subsection{Some results from Gr\"{o}bner basis theory}
Let $\prec$ be a monomial ordering on the set of monomials in $X_1,
\ldots , X_m$. Given an ideal $J \subseteq k[X_1, \ldots , X_m]$,
where $k$ is a field, the footprint $\Delta_{\prec}(J)$ is the set of
monomials that can not be found as leading monomial of any polynomial
in $J$. A Gr\"{o}bner basis, by definition, is a generating set for $J$ from
which the footprint can be easily read of. More formally, $\{L_1(X_1,
\ldots , X_m), \ldots , L_s(X_1, \ldots , X_m)\} \subseteq J$ is a Gr\"{o}bner
basis for $J$ with respect to $\prec$ if for any $F(X_1, \ldots , X_m)
\in J$ for some $i \in \{1, \ldots ,s\}$ it holds that
${\mbox{lm}}(L_i) | {\mbox{lm}}(F)$. Recall that $\{M+J \mid M \in \Delta_\prec(J)\}$ is a basis
for the quotient ring $k[X_1, \ldots , X_m]/J$ as a vector space over
$k$. In the following we shall assume that $k={\mathbb{F}}_q$ and that $J$ contains all the equations
$X_1^q-X_1, \ldots , X_m^q-X_m$, in which case we write
$J=I_q$. Obviously, the variety of $I_q$ is finite. 
Let the variety be $\{P_1, \ldots ,
P_{n}\}$ and consider the evaluation map ${\mbox{ev}}
: {\mathbb{F}}_q[X_1, \ldots , X_m]/I_q \rightarrow
{\mathbb{F}}_q^n$ given by ${\mbox{ev}}(F+I_q)=(F(P_1), \ldots ,
F(P_n))$. It is well-known that this map is a vector space
isomorphism implying that $n=\#   \Delta_\prec(I_q)$ holds. If we embark the 
vector space ${\mathbb{F}}_q^n$ with a second binary operation, namely
the component wise product from Definition~\ref{defcomp} then it becomes an
${\mathbb{F}}_q$-algebra. It is not difficult to see that the map
${\mbox{ev}}$ in this way becomes an isomorphism between
${\mathbb{F}}_q$-algebras. Hence, if we enumerate the elements of
$\Delta_\prec(I_q)=\{ M_1, \ldots , M_n\}$ according to $\prec$ and define
${\mathcal{U}}={\mathcal{V}}={\mathcal{W}}=\{
\vec{b}_1={\mbox{ev}}(M_1+I_q), \ldots ,
\vec{b}_n={\mbox{ev}}(M_n+I_q)\}$ then we can translate information on the algebraic structure of
${\mathbb{F}}_q[X_1, \ldots , X_m]/I_q$ into information regarding
the well-behaving properties as introduced in Definition~\ref{defweb}, \ref{def2},
\ref{def4}, \ref{definvolved} and Proposition~\ref{propihop}. We shall illustrate how to do this in
the following.\\
\subsection{Codes from a curve over ${\mathbb{F}}_8$}
In~\cite[Sec.\ 3.1]{salazar} Salazar et.\ al.\ considered curves of
the form $F_8(X,Y)=G_8(X)-H_8(Y) \in {\mathbb{F}}_8[X,Y]$ where
$G_8(X)$ is a polynomial of degree $4$ and $H_8(Y)$ is a polynomial
of degree $6$ both having the property that when evaluated in
${\mathbb{F}}_8$ they return values in ${\mathbb{F}}_2$. It is of no
implication to the estimation of code parameters if we restrict to
$G_8(X)$ being the trace polynomial $X^4+X^2+X$ and if we choose 
$H_8(Y)=Y^6+Y^5+Y^3$.
Consider the trace-polynomial corresponding to a general field
extension. It is well-known that the preimages of all the elements in the
ground field are of the same size. From this we conclude that the
particular polynomial 
$F_8(X,Y)=G_8(X)-H_8(Y)$ under consideration has exactly $2^5=32$
zeros. 

Let $I_8=\langle F_8(X,Y),X^8-X, Y^8-Y\rangle \subseteq
{\mathbb{F}}_8[X,Y]$. From the above discussion we know that the corresponding variety is of size $32$. If
we consider a monomial ordering such that ${\mbox{lm}}(F_8)=X^4$ then
there exist 
exactly $32$ monomials which are not divisible by any of the monomials
${\mbox{lm}}(F_8)=X^4, {\mbox{lm}}(Y^8-Y)=Y^8$. Hence,
$\{F_8(X,Y),Y^8-Y\}$ is a Gr\"{o}bner basis for $I_8$ and
$\Delta_\prec(I_8)=\{X^\alpha Y^\beta \mid 0 \leq \alpha < 4, 0\leq
\beta < 8\}$ holds. 
In the following we consider a particular weighted degree
lexicographic ordering for which ${\mbox{lm}}(F_8)=X^4$ holds.
Let $w(X)=3$, $w(Y)=2$, and in general
$w(X^\alpha Y^\beta)=3\alpha+2\beta$. We define $\prec_w$ to be the monomial
ordering given by $X^{\alpha_1}Y^{\beta_1} \prec_w
X^{\alpha_2}Y^{\beta_2}$ if either
$w(X^{\alpha_1}Y^{\beta_1})<w(X^{\alpha_2}Y^{\beta_2})$ 
or if
alternatively $w(X^{\alpha_1}Y^{\beta_1})=w(X^{\alpha_2}Y^{\beta_2})$
and $\alpha_1<\alpha_2$ hold.\\

Let $\Delta_{\prec_w}(I_8)=\{M_1, \ldots , M_{32}\}$, the monomials
being enumerated with respect to $\prec_w$. 
For the code construction we consider the
basis ${\mathcal{W}}=\{ \vec{w}_1={\mbox{ev}}(M_1+I_8), \ldots , {\vec{w}}_{32}={\mbox{ev}}(M_{32}+I_8)\}$. The
situation is described in Figure~\ref{figofido}.
\begin{figure}
$$
\begin{array}{ccc}
\begin{array}{c}
\begin{array}{cccc}
Y^7&XY^7&X^2Y^7&X^3Y^7\\
Y^6&XY^6&X^2Y^6&X^3Y^6\\
Y^5&XY^5&X^2Y^5&X^3Y^5\\
Y^4&XY^4&X^2Y^4&X^3Y^4\\
Y^3&XY^3&X^2Y^3&X^3Y^3\\
Y^2&XY^2&X^2Y^2&X^3Y^2\\
Y&XY&X^2Y&X^3Y\\
1&X&X^2&X^3
\end{array}\\
\ \\
{\mbox{Monomials in }} \Delta_{\prec_w}
\end{array}
&
\begin{array}{c}
\begin{array}{rrrr}
14&17&20&23\\
12&15&18&21\\
10&13&16&19\\
8&11&14&17\\
6&9&12&15\\
4&7&10&13\\
2&5&8&11\\
0&3&6&9
\end{array}\\
\ \\
{\mbox{Corresponding weights}}
\end{array}
&
\begin{array}{c}
\begin{array}{rrrr}
21&26&30&32\\
17&23&28&31\\
13&19&25&29\\
9&15&22&27\\
6&11&18&24\\
4&8&14&20\\
2&5&10&16\\
1&3&7&12
\end{array}\\
\ \\
{\mbox{Indexing of ${\mathcal{W}}$}}
\end{array}
\end{array}
$$
\caption{}
\label{figofido}
\end{figure}
We then set $\vec{u}_i=\vec{v}_i=\vec{w}_i$ for $i=1, \ldots , 32$ defining the
bases ${\mathcal{U}}$ and ${\mathcal{V}}$.\\

By definition, $\bar{\rho}_{\mathcal{W}}(\vec{u}_i \ast \vec{v}_j)=l$
if and only if 
$${\mbox{lm}}(M_iM_j {\mbox{ rem }} \{F_8(X,Y),X^8-X,Y^8-Y\})=M_l.$$
Further, $(i,j)$ is WB if and only if 
\begin{equation}
{\mbox{lm}}(M_{i^\prime}M_{j^\prime} {\mbox{ rem }}
\{F_8(X,Y),X^8-X,Y^8-Y\})\prec_w M_l \label{qenabsel}
\end{equation}
holds for all $i^\prime \leq i$ and $j^\prime \leq j$ with
$(i^\prime,j^\prime) \neq (i,j)$. There are two particular easy cases
to analyze:
\begin{itemize}
\item {\bf{Rule (I):}} If $M_iM_j=M_l$ then by the property of a monomial
  ordering~(\ref{qenabsel}) holds.
\item {\bf{Rule (II):}} If $w(M_i)+w(M_j)=w(M_l)$ and $w(M_{i^\prime})<
  w(M_i)$ for all $i^\prime < i$ and if $w(M_{j^\prime})<
  w(M_j)$ for all $j^\prime < j$, then (\ref{qenabsel}) holds. 
\end{itemize}
In a straightforward manner one derives similar rules regarding WWB
and OWB.\\
Consider $l=17$. Using Rule (I) we see that every $$(i,j) \in
\{ (1, 17), (2, 13), (4,9), (6,6), (9,4),(13,2),(17,1)\}$$ is WB with
$\bar{\rho}_{\mathcal{W}}(\vec{u}_i\ast \vec{v}_j)=17$.\\
We have $\bar{\rho}_{\mathcal{W}}(\vec{u}_3 \ast \vec{v}_{12})=17$ as
\begin{eqnarray}
&&{\mbox{lm}}(M_3 M_{12} {\mbox{ rem }}
\{F_8(X,Y),X^8-X,Y^8-Y\})\nonumber \\
&=&{\mbox{lm}}(X^4 {\mbox{ rem }}
\{F_8(X,Y),X^8-X,Y^8-Y\})\nonumber \\
&=&{\mbox{lm}}(Y^6+Y^5+X^2+Y^3+X)=Y^6=M_{17}.\nonumber
\end{eqnarray}
But $M_3M_{11}=M_{18}$ implying that $\bar{\rho}_{\mathcal{W}}(\vec{u}_3 \ast
\vec{v}_{11})=18$. Therefore $(3,12)$ is not WWB. However
$w(M_{i^\prime})< w(M_3)$ for all $i^\prime < 3$ and by a result
similar to  Rule (II), $(3,12)$ therefore is OWB.\\
We next claim that 
${\mathcal{I}}^\prime=\{1, 2, 4, 6, 9, 13, 17, 3, 12\}$ has the
$\mu$-property with respect to $17$. To this end, the only thing
missing to be checked is the case $i=12$. Clearly,
$\bar{\rho}_{\mathcal{W}}(\vec{u}_{12} \ast \vec{v}_3)=17$. Note that
$w(M_{12})=9$ does not belong to $\{w(M_i) \mid i
\in{\mathcal{I}}^\prime \backslash \{ 12\} \}$ and by an argument
similar to Rule (II) we conclude that $(12, 3)$ is OWB with respect to
${\mathcal{I}}^\prime$.\\
We next apply Theorem~\ref{thestrongone} with $l=17$ and $v=1$. Note
that $w(M_{17})=w(M_{18})<w(M_{19})$ which is what makes the choice
$v=1$ natural. Using similar arguments as above we see that 
$${\mathcal{I}}^\prime_0=\{1, 2, 4, 6, 9, 13, 17, 3, 12\}\cup \{7\}$$
has the $\mu$-property with respect to $17$ with exception $\{18\}$
and that 
$${\mathcal{I}}^\prime_1=\{1, 2, 4, 6, 9, 13, 17\} \cup \{ 3, 5, 8, 11\}$$
has the relaxed $\mu$-property with respect to $(17,18)$ with
exception $\{ \}$. Clearly, ${\mathcal{I}}^\prime_0$ is the smallest
of these two sets. \\
In conclusion, if $m(\vec{c})=17$ we get the following estimates:
\begin{itemize}
\item The Feng-Rao bound in the version with WB or WWB gives
  $w_H(\vec{c}) \geq 7$.
\item The same bound in the version with OWB produces $w_H(\vec{c})
  \geq 8$.
\item From the advisory bound we get $w_H(\vec{c}) \geq 9$.
\item Finally, our new bound produces $w_H(\vec{c}) \geq 10$.
\end{itemize}
Applying exactly the same techniques as above we get the following
estimates of $w_H(\vec{c})$ when $m(\vec{c})=21$:
\begin{itemize}
\item The Feng-Rao bound with WB or WWB gives
  $w_H(\vec{c}) \geq 8$.
\item The same bound in the version with OWB produces $w_H(\vec{c})
  \geq 10$.
\item From the advisory bound we get $w_H(\vec{c}) \geq 12$ (This is
  done by choosing ${\mathcal{I}}^\prime=\{1, 2, 4, 6, 9, 13, 17,
  21\}\cup\{3, 5, 12, 16\}$).
\item Finally, our new bound produces $w_H(\vec{c}) \geq 13$ (This
is done by choosing $v=1$, ${\mathcal{I}}^\prime_0=\{1, 2, 4, 6, 9,
13, 17, 21\}\cup\{3, 7, 12, 5, 10, 16\}$ and
${\mathcal{I}}^\prime_1=\{1, 2, 4, 6, 9, 13, 17, 21\}\cup \{3, 5, 8,
11, 15\}$).
\end{itemize}
For the remaining choices of $l\in {\mathcal{I}}$ neither the advisory
bound nor the improved bound from the present paper produces better
results than the Feng-Rao bound with WWB. As explained
in~\cite{salazar} for $m(\vec{c})=28$ and $m(\vec{c})=30$,
respectively, 
the Feng-Rao bound with WWB improves upon the same bound with WB by
lifting the estimates from $21$ to $22$ and from $24$ to $26$, respectively.\\
We first consider the codes $C(s)$ (See Remark~\ref{remstart} for the
definition). In Figure~\ref{figo1} we illustrate the parameters $k$,
 $d_1(C(s)), \ldots , d_5(C(s))$. As is seen, for all of the five
 choices of bounds: the Feng-Rao bound with WB, WWB, OWB, the advisory
 bound, and the bound from Section~\ref{secfurther}, there exist numbers $i$
 and $s$ such that the best estimate on $d_i(C(s))$ is obtained by
 this particular bound (and consequently also by the sharper bounds as
 well). Regarding the $6$th generalized
Hamming weight, only for one $s$ we can improve upon what is derived
from the Feng-Rao bound with WB. Namely, for $C(4)$ where the Feng-Rao bound
with WB or WWB produces the estimate $8$ whereas all other bounds give
$9$. 
\begin{figure}
\begin{center}
\includegraphics[width=128mm]{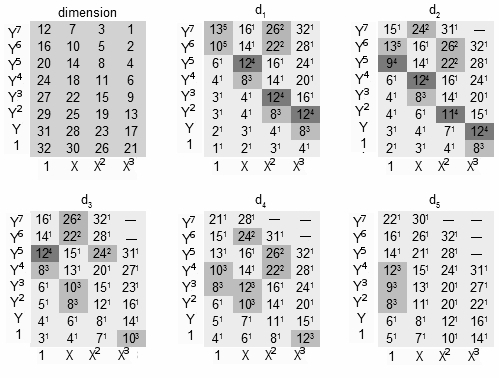}
\end{center}
\caption{The figure lists the dimensions of codes $C(s)$ over
  ${\mathbb{F}}_8$ and corresponding estimates on $d_1, \ldots ,
  d_5$. Information about $C(s)$ is placed at the position of
  $\vec{w}_{s+1}$. An entry $z^1$ means that the value $z$ was obtained from the
  Feng-Rao bound with WB, $z^2$ indicate that the same bound with WWB
  was used, and finally $z^3$ the same bound with OWB. With $z^4$ we
  indicate that the value $z$ was obtained from the advisory bound and
  by $z^5$ that the method from Section~\ref{secfurther} was used. The
symbol - inside the table indicates that the corresponding parameter
does not exist.}
\label{figo1}
\end{figure}
In Table~\ref{tabny} we illustrate that the various bounds
sometimes improve very much on each other by showing estimates for the first two
weights of the code $C(16)$. For this particular code for higher
weights all estimates are the same.
\begin{table}
\begin{center}
\begin{tabular}{l|cccccc}
&Feng-Rao&Feng-Rao&Feng-Rao&Advisory&Section\\
&WB&WWB&OWB&bound&\ref{secfurther}\\
\hline \\
$d_1$&$7$&$7$&$8$&$9$&$10$\\
$d_2$&$8$&$8$&$10$&$12$&$13$
\end{tabular}
\end{center}
\caption{Estimates on first and second generalized Hamming weight of
  the code $C(16)$ over ${\mathbb{F}}_8$.}
\label{tabny}
\end{table}

We next consider the improved codes $\tilde{C}_{adv}(\delta)$ and
$\tilde{C}_{fim}(\delta)$ (See Remark~\ref{remstart} and
Remark~\ref{remslut} for the definitions). For two designed distances
$\delta=10, 13$, the code $\tilde{C}_{fim}(\delta)$ is of higher
dimension than $\tilde{C}_{adv}(\delta)$. In Table~\ref{tabmad} we
list estimates from the advisory bound on the generalized Hamming
weights of the first code and estimates from the bound of
Section~\ref{secfurther} on the generalized Hamming weights of the
latter code, respectively. We see that for higher generalized Hamming weights there
is a price to be paid for the increase in dimension. 
\begin{table}
\begin{center}
\begin{tabular}{l|rrrrrr}
&$k$&$d_2$&$d_3$&$d_4$&$d_5$&$d_6$\\
\hline \\
$\tilde{C}_{adv}(10)$&$16$&$12$&$14$&$15$&$16$&$20$\\
$\tilde{C}_{fim}(10)$&$17$&$12$&$13$&$14$&$15$&$16$\\
$\tilde{C}_{adv}(13)$&$11$&$16$&$20$&$22$&$24$&$26$\\
$\tilde{C}_{fim}(13)$&$12$&$15$&$16$&$21$&$22$&$24$
\end{tabular}
\end{center}
\caption{Parameters of improved codes over ${\mathbb{F}}_8$. By definition, the codes
$\tilde{C}_{adv}(10)$ and $\tilde{C}_{fim}(10)$ are of designed
minimum distance $10$. Similarly,  $\tilde{C}_{adv}(13)$ and
$\tilde{C}_{fim}(13)$,  are of designed
minimum distance $13$. By $k$ we denote the
dimension. The values of $d_2, \ldots
, d_6$ for $\tilde{C}_{adv}(10)$ and $\tilde{C}_{adv}(13)$ are
estimated using the advisory bound. For $\tilde{C}_{fim}(10)$ and
$\tilde{C}_{fim}(13)$ the method from Section~\ref{secfurther} is used.} 
\label{tabmad}
\end{table}
\subsection{Codes from a curve over ${\mathbb{F}}_{27}$}
Similarly to the curve $F_8(X,Y) \in {\mathbb{F}}_8[X,Y]$ from the
previous section we now consider the curve $F_{27}(X,Y)=G_{27}(X)-H_{27}(Y)\in {\mathbb{F}}_{27}[X,Y]$. Here, $G_{27}(X)$ is
the trace-polynomial $X^9+X^3+X$ and $H_{27}(Y)=Y^{12}+Y^{10}+Y^4$ 
satisfies that when evaluated in elements from ${\mathbb{F}}_{27}$ it
returns values from ${\mathbb{F}}_3$. The arguments of the previous
subsection translate immediately. Only difference is that now instead of
having many pairs of monomials in the footprint being of the same
weight we now have many triples of monomials in the footprint being
of the same weight. The implication is that when applying
Theorem~\ref{thenew} we will often need $v=2$ rather than $v=1$. The codes being of length $n=3^5=243$
we cannot give many details, but restrict to consider the minimum
distance and the second generalized Hamming
weight of the codes $C(s)$. See Figure~\ref{figo2}. Again, all five
bounds come into action.  
\begin{figure}
\begin{center}
\includegraphics[height=86mm]{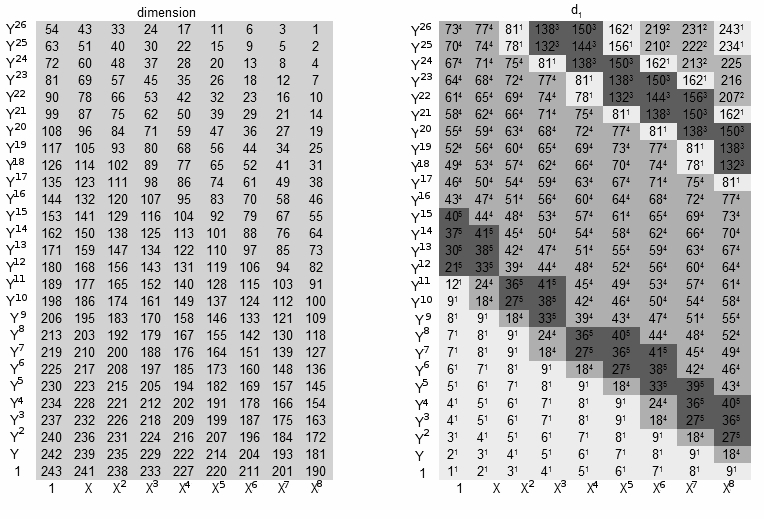}\\
\includegraphics[height=86mm]{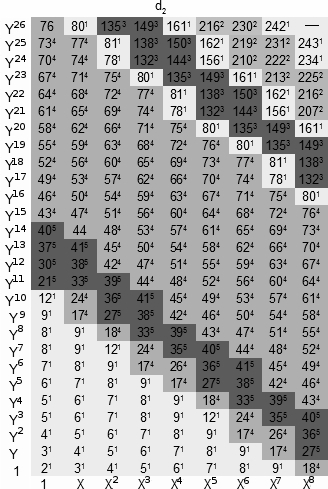}
\end{center}
\caption{Dimensions, minimum distance and second generalized
  Hamming weight of codes $C(s)$ over
  ${\mathbb{F}}_{27}$. Notation as in Figure~\ref{figo1}}
\label{figo2}
\end{figure}
To illustrate how much the advisory bound and the bound
of Section~\ref{secfurther} improve upon the various versions of the
Feng-Rao bound we treat in detail the codes $C(75)$, $C(76)$, $C(83)$
in Table~\ref{taberitop}. These codes are of dimension
$168$, $167$ and $160$. 
\begin{table}
\begin{center}
\begin{tabular}{l|cccccc}
&Feng-Rao&Feng-Rao&Feng-Rao&Advisory&Section\\
&WB&WWB&OWB&bound&\ref{secfurther}\\
\hline \\
$d_1(C(75))$&$15$&$15$&$21$&$29$&$33$\\
$d_2(C(75))$&$16$&$16$&$24$&$34$&$38$\\
 &     \\
$d_1(C(76))$&$15$&$15$&$21$&$33$&$36$\\
$d_2(C(76))$&$16$&$16$&$24$&$38$&$39$\\
& \\
$d_1(C(83))$&$16$&$16$&$24$&$34$&$38$\\
$d_2(C(83))$&$17$&$17$&$27$&$39$&$41$
\end{tabular}
\end{center}
\caption{Estimates of minimum distance and second generalized Hamming
  weight for a selection of codes over ${\mathbb{F}}_{27}$.}
\label{taberitop}
\end{table}

\section{Concluding remarks}
In this paper we treated two improvements to the Feng-Rao bound for
dual codes: the advisory bound and a new bound which is an improvement
to it. The latter bound is closely related to a new bound for primary
codes which we treat in a separate paper. 
Part of this research
was done while the second listed author was visiting East China
Normal University. We are grateful to Professor Hao Chen for his hospitality. The authors also gratefully acknowledge the support from
the Danish National Research Foundation and the National Science
Foundation of China (Grant No.\ 11061130539) for the Danish-Chinese
Center for Applications of Algebraic Geometry in Coding Theory and
Cryptography.

\bibliography{bibfile}

\begin{thebibliography}{10}

\bibitem{AG}
Henning~E. Andersen and Olav Geil.
\newblock Evaluation codes from order domain theory.
\newblock {\em Finite Fields Appl.}, 14(1):92--123, 2008.

\bibitem{FR1}
Gui~Liang Feng and T.~R.~N. Rao.
\newblock A simple approach for construction of algebraic-geometric codes from
  affine plane curves.
\newblock {\em IEEE Trans. Inform. Theory}, 40(4):1003--1012, 1994.

\bibitem{FR2}
Gui-Liang Feng and T.~R.~N. Rao.
\newblock Improved geometric {G}oppa codes part {I}: Basic theory.
\newblock {\em IEEE Trans. Inform. Theory}, 41(6):1678--1693, 1995.

\bibitem{lax}
J.~Fitzgerald and R.~F. Lax.
\newblock Decoding affine variety codes using {G}r\"obner bases.
\newblock {\em Des. Codes Cryptogr.}, 13(2):147--158, 1998.

\bibitem{GP}
Olav Geil and Ruud Pellikaan.
\newblock On the structure of order domains.
\newblock {\em Finite Fields Appl.}, 8(3):369--396, 2002.

\bibitem{geithom}
Olav Geil and Christian Thommesen.
\newblock On the {F}eng-{R}ao bound for generalized {H}amming weights.
\newblock In Marc~P.C. Fossorier, Hideki Imai, Shu Lin, and Alain Poli,
  editors, {\em Applied Algebra, Algebraic Algorithms and Error-Correcting
  Codes}, volume 3857 of {\em Lecture Notes in Computer Science}, pages
  295--306. Springer, 2006.

\bibitem{heijnenpellikaan}
Petra Heijnen and Ruud Pellikaan.
\newblock Generalized {H}amming weights of {$q$}-ary {R}eed-{M}uller codes.
\newblock {\em IEEE Trans. Inform. Theory}, 44(1):181--196, 1998.

\bibitem{handbook}
Tom H{\o}holdt, Jacobus~H. van Lint, and Ruud Pellikaan.
\newblock Algebraic geometry codes.
\newblock In Vera~S. Pless and William~Cary Huffman, editors, {\em Handbook of
  Coding Theory}, volume~1, pages 871--961. Elsevier, Amsterdam, 1998.

\bibitem{kurihara}
J.~Kurihara, T.~Uyematsu, and R.~Matsumoto.
\newblock Secret sharing schemes based on linear codes can be precisely
  characterized by the relative generalized hamming weight.
\newblock {\em IEICE Trans. Fundamentals}, E95-A(11):2067--2075, 2012.

\bibitem{luoetal}
Zihui Liu, Wende Chen, and Yuan Luo.
\newblock The relative generalized {H}amming weight of linear {$q$}-ary codes
  and their subcodes.
\newblock {\em Des. Codes Cryptogr.}, 48(2):111--123, 2008.

\bibitem{luo}
Y.~Luo, C.~Mitrpant, A.J.H. Vinck, and K.~Chen.
\newblock Some new characters on the wire-tap channel of type ii.
\newblock {\em Information Theory, IEEE Transactions on}, 51(3):1222--1229,
  2005.

\bibitem{MM}
Ryutaroh Matsumoto and Shinji Miura.
\newblock On the {F}eng-{R}ao bound for the $\mathcal{L}$-construction of
  algebraic geometry codes.
\newblock {\em IEICE Trans. Fundamentals}, E83-A(5):926--930, May 2000.

\bibitem{miura1}
Shinji Miura.
\newblock {\em Study of Error-Correcting Codes based on Algebraic Geometry}.
\newblock PhD thesis, Univ. Tokyo, 1997.
\newblock (in Japanese).

\bibitem{early}
Ruud Pellikaan.
\newblock On the efficient decoding of algebraic-geometric codes.
\newblock In P.~Camion, P.~Charpin, and S.~Harari, editors, {\em Eurocode '92
  International Symposium on Coding Theory and Applications}, number 339 in
  CISM Courses and Lectures, pages 231--253. CISM International Centre for
  Mechanical Sciences, Springer, 1993.

\bibitem{salazar}
G.~Salazar, D.~Dunn, and S.~B. Graham.
\newblock An improvement of the {F}eng-{R}ao bound on minimum distance.
\newblock {\em Finite Fields Appl.}, 12:313--335, 2006.

\bibitem{wei}
V.K. Wei.
\newblock Generalized hamming weights for linear codes.
\newblock {\em Information Theory, IEEE Transactions on}, 37(5):1412--1418,
  1991.

\end{thebibliography}
\bibliographystyle{plain}

\end{document}